\documentclass[a4paper,reqno]{amsart}
\usepackage{amsmath,amssymb,mathrsfs}
\usepackage{txfonts,bm,pifont}
\usepackage{cite}
\usepackage[dvipdfm,bookmarks=true,bookmarksnumbered=true,colorlinks=true]{hyperref}
\usepackage{comment}
\usepackage[dvipdfmx]{graphicx}
\PassOptionsToPackage{svgnames}{xcolor}
\usepackage{tikz}
\usetikzlibrary{calc}
\usetikzlibrary{decorations.markings,decorations.pathmorphing}

\newtheorem{theorem}{Theorem}[section]

\newtheorem{remark}[theorem]{Remark}
\newtheorem{definition}[theorem]{Definition}
\numberwithin{equation}{section}
\newcommand{\sn}[1]{{\rm sn}\left(#1\right)}
\newcommand{\cn}[1]{{\rm cn}\left(#1\right)}
\newcommand{\dn}[1]{{\rm dn}\left(#1\right)}
\newcommand{\cd}[1]{{\rm cd}\left(#1\right)}
\newcommand{\JT}[1]{\Theta\left(#1\right)}
\newcommand{\JH}[1]{H\left(#1\right)}
\newcommand{\JTK}[1]{\Theta_1\left(#1\right)}
\newcommand{\JHK}[1]{H_1\left(#1\right)}
\newcommand{\ssn}[1]{{\rm sn}^2\left(#1\right)}
\newcommand{\ccn}[1]{{\rm cn}^2\left(#1\right)}
\newcommand{\ddn}[1]{{\rm dn}^2\left(#1\right)}
\newcommand{\we}[1]{\wp\left(#1\right)}
\newcommand{\pP}[2]{P_{#1}\left(#2\right)}
\newcommand{\pQ}[2]{Q_{#1}\left(#2\right)}
\newcommand{\gG}[1]{G_{#1}\,}
\newcommand{\bfrac}[2]{\left(\frac{#1}{#2}\right)}
\newcommand{\bcfrac}[2]{\left(\cfrac{#1}{#2}\right)}
\newcommand{\lrangle}[1]{\langle#1\rangle}
\newcommand{\ox}{\overline{x}}
\newcommand{\oy}{\overline{y}}
\newcommand{\tx}{\tilde{x}}
\newcommand{\ttx}{\tilde{\tx}}
\newcommand{\ty}{\tilde{y}}
\newcommand{\la}{\lambda}
\newcommand{\ka}{\kappa}

\newcommand{\ii}{{\rm i}}
\newcommand{\Th}{\Theta}
\newcommand{\ep}{\epsilon}
\newcommand{\de}{\delta}
\newcommand{\al}{\alpha}
\newcommand{\io}{\iota}
\newcommand{\bbP}{\mathbb{P}}
\newcommand{\bbZ}{\mathbb{Z}}
\newcommand{\PicX}{{\rm Pic}(X)}
\newcommand{\WEe}{W\big(E_8^{(1)}\big)}
\newcommand{\tWEe}{\widetilde{W}\big(E_8^{(1)}\big)}
\newcommand{\aaa}[1]{a_{#1}}
\newcommand{\cc}[1]{c_{#1}}
\newcommand{\EE}[1]{E_{#1}}
\newcommand{\TJ}[1]{T_{J,#1}}
\newcommand{\RJ}[1]{R_{J,#1}}
\newcommand{\gae}{\gamma_{\rm e}}
\newcommand{\gao}{\gamma_{\rm o}}
\newcommand{\sz}{{\rm sz}}
\newcommand{\cz}{{\rm cz}}
\newcommand{\dz}{{\rm dz}}
\newcommand{\hsz}{{\rm \widehat{sz}}}
\newcommand{\hcz}{{\rm \widehat{cz}}}
\newcommand{\hdz}{{\rm \widehat{dz}}}
\newcommand{\sge}{{\rm sg}_{\rm e}}
\newcommand{\cge}{{\rm cg}_{\rm e}}
\newcommand{\dge}{{\rm dg}_{\rm e}}
\newcommand{\sgo}{{\rm sg}_{\rm o}}
\newcommand{\cgo}{{\rm cg}_{\rm o}}
\newcommand{\dgo}{{\rm dg}_{\rm o}}
\makeatletter
\long\def\@makecaption#1#2{
 \vskip 10pt
 \setbox\@tempboxa\hbox{#1. #2}
 \ifdim \wd\@tempboxa >\hsize #1. #2\par \else \hbox
to\hsize{\hfil\box\@tempboxa\hfil}
 \fi}
\begin{document}
\title{Elliptic Painlev\'e equations from next-nearest-neighbor translations on the $E_8^{(1)}$ lattice}
\author{Nalini Joshi}
\author{Nobutaka Nakazono}
\address{School of Mathematics and Statistics F07, The University of Sydney, New South Wales 2006, Australia.}
\email{nobua.n1222@gmail.com}
\begin{abstract}
The well known elliptic discrete Painlev\'e equation of Sakai is constructed by a standard translation on the $E_8^{(1)}$ lattice, given by nearest neighbor vectors. 
In this paper, we give a new elliptic discrete Painlev\'e equation obtained by translations along next-nearest-neighbor vectors. 
This equation is a generic (8-parameter) version of a 2-parameter elliptic difference equation found by reduction from Adler's partial difference equation, the so-called Q4 equation. 
We also provide a projective reduction of the well known equation of Sakai. 
\end{abstract}
\maketitle
\setcounter{tocdepth}{1}

\section{Introduction}
\label{section:introduction}
At the head of the list of discrete Painlev\'e equations described by Sakai sits an elliptic difference equation, which has attracted a great deal of attention in recent times. 
This equation \cite{SakaiH2001:MR1882403,MSY2003:MR1958273} has the affine Weyl group symmetry of type $E_8^{(1)}$ and a perennial question is whether it is unique as the only elliptic-difference-type discrete Painlev\'e equation. 
While two other candidate equations are now known \cite{ORG2001:MR1877472,RCG2009:MR2525848}, these are related to Sakai's elliptic difference equation. 
In this paper, we deduce elliptic-difference equations for the first time, which are not related to Sakai's equation by either a Miura transformation or by projective reduction. 
We do this by considering different (non-conjugate) translations on the $E_8^{(1)}$ lattice and describing these translations under the action of Jacobian elliptic functions.

A discrete Painlev\'e equation is an ordinary difference equation, which is iterated by translation in an affine Weyl group, where translation corresponds to vectors in the root- or (equivalently) weight-lattice \cite{book_CSBBLNOPQV2013:Sphere,book_HumphreysJE1992:Reflection}.
Sakai's elliptic difference equation is realized by translation on the $E_8^{(1)}$ lattice.
 
To describe the construction, we first explain how such translations are characterized. 
Fix a point in the $E_8^{(1)}$ lattice. 
Then there are 240 nearest neighbors of this point in the lattice, lying at a distance whose squared length is equal to 2. 
We refer to the 120 vectors between this fixed point and its possible nearest neighbors as nearest-neighbor-connecting vectors (NVs).
Similarly, there are 2160 next-nearest neighbors, lying at a distance whose squared length is 4.
The 1080 vectors between the fixed point and such next-nearest neighbors will be referred to as next-nearest-neighbor-connecting vectors (NNVs). 
In \cite{SakaiH2001:MR1882403} and \cite{ORG2001:MR1877472}, elliptic difference equations were constructed as translations expressed in terms of NVs. 

In a previous study \cite{AHJN2016:MR3509963}, we considered another elliptic difference equation (Equation \eqref{eqns:RCGeqn}), which was obtained by reduction from a partial difference equation \cite{RCG2009:MR2525848}. 
Curiously, its symmetry group turns out to be $W(F_4^{(1)})$, a sub-group of $W(E_8^{(1)})$. 
Moreover, its time iteration turns out not to be given by translation on the $E_8^{(1)}$ lattice. 
However, its square (i.e., composition with itself) {\em is} a translation. 
We have discovered that this translation is expressed in terms of NNVs, which makes it very different to the elliptic difference equations in earlier papers \cite{SakaiH2001:MR1882403,ORG2001:MR1877472}.
Note that independent study by Carstea {\it et al.} \cite{CDT2017:arXiv170204907} has also led to equations with translations given by NNVs.

Generalizing this surprising insight led us to the discovery of a new elliptic Painlev\'e equation \eqref{eqns:TJ1}, which can be regarded as the generic version of Equation \eqref{eqns:RCGeqn}. 
The term generic here refers to the fact that it contains the largest number of parameters possible as an equation with symmetry group $W(E_8^{(1)})$. 
It has 8 parameters in addition to the independent variable. 
Moreover, we also obtain a projective reduction of Sakai's elliptic difference equation \eqref{eqn:MSY_dP_PR}, which contains 4 parameters in addition to the independent variable.

All of these elliptic difference equations have the same space of initial values, regularized by blowing up 8 points in arbitrary position in $\bbP^1\times\bbP^1$. 
The space of initial values is labelled as $A_0^{(1)}$. 
We emphasize here that this characterization of the space of initial values is not enough to distinguish the four different elliptic difference equations described above. 
We provide evidence that different (non-conjugate) translations in $\tWEe$ lead to distinct elliptic difference equations. 
All four elliptic difference equations are described in this paper in terms of Jacobian elliptic functions for the sake of uniformity, because this is the form originally given for Equation \eqref{eqns:RCGeqn}. 

\subsection{Background}
Discrete Painlev\'e equations are nonlinear ordinary difference equations of second order, 
which include discrete analogues of the six Painlev\'e equations: P$_{\rm I}$, $\dots$, P$_{\rm VI}$. 
Together with the Painlev\'e equations, the discrete Painlev\'e equations are now regarded as one of the most important classes of equations in the theory of integrable systems (see, e.g., \cite{GR2004:MR2087743}). 

Sakai's geometric description of discrete Painlev\'e equations, based on types of space of initial values,
is well known\cite{SakaiH2001:MR1882403}.
This picture relies on compactifying and regularizing space of initial values.
The spaces of initial values are constructed by the blow up of $\bbP^1\times\bbP^1$ at eight base points
(i.e. points where the system is ill defined because it approaches $0/0$)
and are classified into 22 types according to the configuration of the base points as follows:
\begin{center}
\begin{tabular}{|l|l|}
\hline
Discrete type&Type of surface\\
\hline
Elliptic&$A_0^{(1)}$\rule[-.5em]{0em}{1.6em}\\
\hline
Multiplicative&$A_0^{(1)\ast}$, $A_1^{(1)}$, $A_2^{(1)}$, $A_3^{(1)}$, \dots, $A_8^{(1)}$, $A_7^{(1)'}$\rule[-.5em]{0em}{1.6em}\\
\hline
Additive&$A_0^{(1)\ast\ast}$, $A_1^{(1)\ast}$, $A_2^{(1)\ast}$, $D_4^{(1)}$, \dots, $D_8^{(1)}$, $E_6^{(1)}$, $E_7^{(1)}$, $E_8^{(1)}$\rule[-.5em]{0em}{1.6em}\\
\hline
\end{tabular}
\end{center}
In each case, the root system characterizing the surface forms a subgroup of the 10-dimensional Picard group. 
The symmetry group of each equation, formed by Cremona isometries, arises from the orthogonal complement of this root system inside the Picard group.
Its birational actions give the discrete Painlev\'e equation of interest in each case. 

Recently, the following elliptic Painlev\'e equation was obtained from the reduction of Adler's partial difference equation (or, Q4 equation)\cite{RCG2009:MR2525848}:
\begin{subequations}\label{eqns:RCGeqn}
\begin{align}
 &\ty=\frac{(1-k^2\sz^4)\cge\dge\, xy-(\cge^2-\cz^2)\cz\,\dz-(1-k^2\sge^2 \sz^2)\cz\,\dz\, x^2}
 {k^2(\cge^2-\cz^2)\cz\,\dz\, x^2 y-(1-k^2 \sz^4)\cge\dge\, x+(1-k^2\sge^2\sz^2)\cz\,\dz\,y},\\
 &\tx=\frac{(1-k^2\hsz^4)\cgo\dgo\, \ty x-(\cgo^2-\hcz^2)\hcz\,\hdz-(1-k^2\sgo^2 \hsz^2)\hcz\,\hdz\,\ty^2}
 {k^2(\cgo^2-\hcz^2)\hcz\,\hdz\, \ty^2x-(1-k^2 \hsz^4)\cgo\dgo\, \ty+(1-k^2\sgo^2 \hsz^2)\hcz\,\hdz\, x},
\end{align}
\end{subequations}
where $k$ is the modulus of the elliptic sine,
\begin{subequations}
\begin{align}
 &\sz=\sn{z_0},\quad
 \hsz=\sn{z_0+\gae+\gao},
 &&\sge=\sn{\gae},\quad
 \sgo=\sn{\gao},\\
 &\cz=\cn{z_0},\quad
 \hcz=\cn{z_0+\gae+\gao},
 &&\cge=\cn{\gae},\quad
 \cgo=\cn{\gao},\\
 &\dz=\dn{z_0},\quad
 \hdz=\dn{z_0+\gae+\gao},
 &&\dge=\dn{\gae},\quad
 \dgo=\dn{\gao},
\end{align}
\end{subequations}
and 
\begin{equation}
 \tilde{}\,:(\gae,\gao,z_0)\mapsto \big(\gae,\gao,z_0+2(\gae+\gao)\big).
\end{equation}
See Appendix \ref{section:Appendix_A} for standard results about Jacobian elliptic functions.

In \cite{AHJN2016:MR3509963}, the geometry of Equation \eqref{eqns:RCGeqn}, i.e., its space of initial values and corresponding Cremona isometries, was investigated. 
The space of initial values of Equation \eqref{eqns:RCGeqn} was identified with the elliptic $A_0^{(1)}$-surface and its Cremona isometries collectively form an affine Weyl transformation group of type $E_8^{(1)}$, denoted by $\WEe$.
Moreover, it was shown that Equation \eqref{eqns:RCGeqn} cannot be derived from a translation of $\WEe$ but can be derived by using a projective reduction.
The process of deriving discrete dynamical systems of Painlev\'e type from elements of affine Weyl groups that are of infinite order (but that are not necessarily translations) by taking a projection on an appropriate subspace of parameters is called a projective reduction \cite{KNT2011:MR2773334,KN2015:MR3340349}.
Note that this process is motivated by taking symmetric versions of systems of discrete Painlev\'e equations, which has been known since \cite{RGH1991:MR1125951}.

Although the geometry of Equation \eqref{eqns:RCGeqn} has been clarified, 
the realization of Equation \eqref{eqns:RCGeqn} from the action of Cremona isometries was missing
since its base points are parametrized by the Jacobian elliptic function (Jacobi's setting)
and birational actions of Cremona isometries on such setting were not explicitly known. 

The present study fills this gap (see Theorem \ref{theorem:birational_tWE8}), that is, we provide the realization of the Equation \eqref{eqns:RCGeqn}
as a half-translation of the extended affine Weyl group $\tWEe$.

In the remainder of the paper, we refer to Sakai's elliptic Painlev\'e equation \cite{SakaiH2001:MR1882403}
as the MSY elliptic Painlev\'e equation \cite{MSY2003:MR1958273} because the former was obtained in $\bbP^2$
while the latter was provided by resolution in $\bbP^1\times\bbP^1$.
We work throughout the paper in $\bbP^1\times\bbP^1$.

\subsection{Plan of the paper}
This paper is organized as follows. 
In \S \ref{section:cremona}, we construct Cremona isometries for base points \eqref{eqn:basepoints} that generalize those of Equation \eqref{eqns:RCGeqn} and show that these form an affine Weyl group of type $\EE8^{(1)}$, denoted by $\WEe$, under the linear actions on the Picard group. 
In \S \ref{section:birational}, using intersection theory, we obtain the birational action of $\WEe$ on the coordinates and parameters of the base points, and we prove that these birational actions also satisfy the fundamental relations of an affine Weyl group of type $\EE8^{(1)}$.
Adding the identity mappings on the Picard group, we obtain the extended affine Weyl group $\tWEe$. 
Equation \eqref{eqns:RCGeqn} and three other elliptic Painlev\'e equations are then deduced from the action of this resulting group $\tWEe$. 
Finally, we give some concluding remarks in \S \ref{ConcludingRemarks}.
\section{Cremona isometries}
\label{section:cremona}

In this section, we construct Cremona isometries for a generalization of the base points of Equation \eqref{eqns:RCGeqn} and show that these isometries collectively form an affine Weyl group of type $\EE8^{(1)}$.

Recall that Equation \eqref{eqns:RCGeqn} has the following eight base points (see \cite{AHJN2016:MR3509963}):
\begin{subequations}\label{eqn:basepoints_RCG}
\begin{align}
 &p_1:(x,y)=\big(\cd{\gao+\ka},\cd{z_0-\gae-\gao+\ka}\big),\\
 &p_2:(x,y)=\big(\cd{\gao+\ii K'},\cd{z_0-\gae-\gao+\ii K'}\big),\\
 &p_3:(x,y)=\big(\cd{\gao+2K},\cd{z_0-\gae-\gao+2K}\big),\\
 &p_4:(x,y)=\big(\cd{\gao},\cd{z_0-\gae-\gao}\big),\\
 &p_5:(x,y)=\big(\cd{z_0+\ka},\cd{\gae+\ka}\big),\\
 &p_6:(x,y)=\big(\cd{z_0+\ii K'},\cd{\gae+\ii K'}\big),\\
 &p_7:(x,y)=\big(\cd{z_0+2K},\cd{\gae+2K}\big),\\
 &p_8:(x,y)=\big(\cd{z_0},\cd{\gae}\big),
\end{align}
\end{subequations}
where $K=K(k)$ and $K'=K'(k)$ are complete elliptic integrals and
\begin{equation}\label{eqn:kappa}
 \ka=2K+\ii K',
\end{equation}
which lie on the elliptic curve
\begin{equation}\label{eqn:RCG_curve}
 \sn{z_0-\gae}^2(1+k^2x^2y^2)+2\cn{z_0-\gae}\dn{z_0-\gae}xy-(x^2+y^2)=0.
\end{equation}
Note that these base points contain three arbitrary parameters $z_0$, $\gae$, $\gao$.

We now generalize the number of arbitrary parameters. 
Take a new set of base points given by
\begin{equation}\label{eqn:basepoints}
 p_i:(x,y)=\big(\cd{c_i+\eta},\cd{\eta-c_i}\big),\quad i=1,\dots,8,
\end{equation}
where $c_i$, $i=1,\dots,8$, and $\eta$ are non-zero complex parameters. 
These generalized base points lie on the elliptic curve
\begin{equation}\label{eqn:Jelliptic_curve}
 \sn{2\eta}^2(1+k^2x^2y^2)+2\cn{2\eta}\dn{2\eta}xy-(x^2+y^2)=0.
\end{equation}
These base points can be reduced to the ones given in Equation \eqref{eqn:basepoints_RCG} under a specialization of the parameters.
Indeed, the points \eqref{eqn:basepoints} and elliptic curve \eqref{eqn:Jelliptic_curve}
can be respectively reduced to the points \eqref{eqn:basepoints_RCG} and the curve \eqref{eqn:RCG_curve} by assuming
\begin{subequations}\label{eqns:condition_RCG}
\begin{align} 
 &\cc2=\cc1+2K,\quad
 \cc3=\cc1+\ii K',\quad
 \cc4=\cc1+\ka,\\
 &\cc6=\cc5+2K,\quad
 \cc7=\cc5+\ii K',\quad
 \cc8=\cc5+\ka,
\end{align}
and letting
\begin{equation}
 z_0=\eta+\cc5+\ka,\quad
 \gae=\cc5-\eta+\ka,\quad
 \gao=\eta+\cc1+\ka.
\end{equation}
\end{subequations}

\begin{remark}
To distinguish between Weierstrass's and Jacobi's setting,
we here denote the surface characterized by the elliptic curve \eqref{eqn:Jelliptic_curve} as $A_0^{(1)J}$-surface.
The relation between $A_0^{(1)}$- and $A_0^{(1)J}$-surface is given in Appendix \ref{section:Appendix_B}.
\end{remark}
Let $\ep: X \to \bbP^1\times\bbP^1$ denote the blow up of $\bbP^1\times\bbP^1$ at the points \eqref{eqn:basepoints}. Moreover, let the linear equivalence classes of the total transform of vertical and horizontal lines in $\bbP^1\times\bbP^1$ be denoted respectively by $H_0$ and $H_1$.

The Picard group of $X$, denoted by Pic$(X)$, is given by
\begin{equation}
 \PicX=\bbZ H_0\bigoplus\bbZ H_1\bigoplus_{i=1}^8\bbZ \EE{i},
\end{equation}
where $\EE{i}=\ep^{-1}(p_i)$, $i=1,\dots,8$, are exceptional divisors. 
The intersection form $(\,|\,)$ is given by a symmetric bilinear form with 
\begin{equation}
 (H_i|H_j)=1-\de_{ij},\quad
 (H_i|\EE{j})=0,\quad
 (\EE{i}|\EE{j})=-\de_{ij},
\end{equation}
where $\de_{ij}$ is the Kronecker delta.
The anti-canonical divisor of $X$ is given by 
\begin{equation}
 -K_X=2H_0+2H_1-\sum_{i=1}^8\EE{i}.
\end{equation}
For later convenience, let
\begin{equation}
 \delta=-K_X.
\end{equation}

We define the root lattice $Q(A_0^{(1)\bot})=\bigoplus_{i=0}^8\bbZ\alpha_i$ by the elements of Pic$(X)$ 
that are orthogonal to the anti-canonical divisor $\delta$. 
The simple roots $\alpha_i$, $i=0,\dots,8$, are given by
\begin{subequations}\label{eqns:simpleroots_alpha}
\begin{align}
 &\alpha_1=H_1-H_0,\quad
 \alpha_2=H_0-\EE1-\EE2,\quad
 \alpha_i=\EE{i-1}-\EE{i},\quad i=3,\dots,7,\\
 &\alpha_8=\EE1-\EE2,\quad
 \alpha_0=\EE7-\EE8,
\end{align}
\end{subequations}
where
\begin{equation}
 \delta=2\alpha_1+4\alpha_2+6\alpha_3+5\alpha_4+4\alpha_5+3\alpha_6+2\alpha_7+3\alpha_8+\alpha_0.
\end{equation}
We can easily verify that
\begin{equation}
 (\alpha_i|\alpha_j)
 =\begin{cases}
 -2,& i=j\\
 \ 1, &i=j-1\quad (j=2,\dots,7),\quad \text{or\hspace{0.5em}if}\quad (i,j)=(3,8),(7,0)\\
 \ 0, &\text{otherwise}.
\end{cases}
\end{equation}
Representing intersecting $\alpha_i$ and $\alpha_j$ by a line between nodes $i$ and $j$, we obtain the Dynkin diagram of $E_8^{(1)}$ shown in Figure \ref{fig:dynkin_E8}.

\begin{figure}[htbp]
\begin{center}
\includegraphics[width=0.5\textwidth]{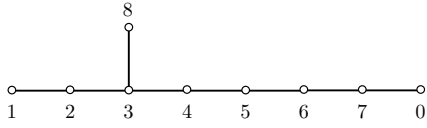}
\caption{Dynkin diagram for the root lattice $\bigoplus_{i=0}^8\mathbb{Z}\alpha_i$.}
\label{fig:dynkin_E8}
\end{center}
\end{figure}

\begin{definition}{\rm \cite{DolgachevIV2012:MR2964027,DO1988:MR1007155,LooijengaE1981:MR632841}}\label{def:cremona}
An automorphism of Pic$(X)$ is called a Cremona isometry if it preserves 
\begin{description}
\item[(i)]
the intersection form $(\,|\,)$ on Pic$(X)$;
\item[(ii)]
the canonical divisor $K_X$;
\item[(iii)]
effectiveness of each effective divisor of Pic$(X)$.
\end{description}
\end{definition}

It is well-known that the reflections are Cremona isometries.
In this case we define the reflections $s_i$, $i=0,\dots,8$, by the following linear actions:
\begin{equation}\label{eqn:def_si}
 s_i.v=v-\cfrac{2(v|\al_i)}{(\al_i|\al_i)}\,\al_i,
\end{equation}
for all $v\in \PicX$.

They collectively form an affine Weyl group of type $\EE8^{(1)}$, denoted by $\WEe$. 
Namely, we can easily verify that under the action on the $\PicX$ the following fundamental relations hold:
\begin{equation}\label{eqns:fundamental_We8}
 (s_is_j)^{l_{ij}}=1,
\end{equation}
where 
\begin{equation}
l_{ij}=
\begin{cases}
 1,& i=j\\
 3, &i=j-1\quad (j=2,\dots,7),\quad \text{or\hspace{0.5em}if}\quad (i,j)=(3,8),(7,0)\\
 2, &\text{otherwise}.
\end{cases}
\end{equation}

\section{Birational actions of the Cremona isometries for the Jacobi's setting}
\label{section:birational}
In this section, we give the birational actions of the Cremona isometries on the coordinates and parameters of the base points \eqref{eqn:basepoints}.
By using these birational actions, we derive various elliptic Painlev\'e equations. 

We focus on a particular example first to explain how to deduce such birational actions. 
Recall $H_0$ and $H_1$ are given by the linear equivalence classes of vertical lines $x=\text{constant}$ and horizontal lines $y=\text{constant}$, respectively. 
Applying the reflection operator $s_2$ given by \eqref{eqn:def_si} to $H_1$, we find
\begin{equation}
 s_2.H_1=H_0+H_1-E_1-E_2,
\end{equation}
which means that $s_2(y)$ can be described by the curve of bi-degree $(1,1)$ passing through base points $p_1$ and $p_2$ with multiplicity $1$.
 (See \cite{KNY2015:arXiv150908186K} for for more detail.)
This result leads us to the birational action given below in Equation \eqref{eqn:action_E8_J_y}.

Similarly, from the linear actions of $s_i$, $i=0,\dots,8$, we obtain their birational actions on the coordinates and parameters of the base points \eqref{eqn:basepoints} as follows.
The actions of the generators of $\WEe$ on the coordinates $(x,y)$ are given by
\begin{subequations}\label{eqns:action_E8_J_para_xy}
\begin{align}
 &s_1(x)=y,\quad
 s_1(y)=x,\\
 &\bfrac{s_2(y)-\cd{2\eta-\frac{\cc1-\cc2}{2}}}{s_2(y)-\cd{2\eta+\frac{\cc1-\cc2}{2}}}
 \bfrac{x-\cd{\eta+\cc1}}{x-\cd{\eta+\cc2}}
 \bfrac{y-\cd{\eta-\cc2}}{y-\cd{\eta-\cc1}}\notag\\
 &\qquad =\bfrac{1-\frac{\cd{\eta-\cc2}}{\cd{\eta}}}{1-\frac{\cd{\eta-\cc1}}{\cd{\eta}}}
 \bfrac{1-\frac{\cd{\eta+\cc1}}{\cd{\eta}}}{1-\frac{\cd{\eta+\cc2}}{\cd{\eta}}}
 \bfrac{1-\frac{\cd{2\eta-\frac{\cc1-\cc2}{2}}}{\cd{\frac{\cc1+\cc2}{2}}}}{1-\frac{\cd{2\eta+\frac{\cc1-\cc2}{2}}}{\cd{\frac{\cc1+\cc2}{2}}}},
 \label{eqn:action_E8_J_y}
\end{align}
while those on the parameters $\cc{i}$, $i=1,\dots,8$, and $\eta$ are given by
\begin{align}
 &s_0(\cc7)=\cc8,\quad
 s_0(\cc8)=\cc7,\quad
 s_1(\eta)=-\eta,\\
 &s_2(\eta)=\eta-\frac{2\eta+\cc1+\cc2}{4},\quad
 s_2(\cc{i})=
 \begin{cases}
 \cc{i}-\dfrac{3(2\eta+\cc1+\cc2)}{4},& i=1,2\\
 \cc{i}+\dfrac{2\eta+\cc1+\cc2}{4},& i\neq 1,2,
 \end{cases}\\
 &s_k(\cc{k-1})=\cc{k},\quad 
 s_k(\cc{k})=\cc{k-1},\quad
 k=3,\dots,7,\\
 &s_8(\cc1)=\cc2,\quad
 s_8(\cc2)=\cc1.
\end{align}
\end{subequations}
Note that 
\begin{equation}
 \la=\sum_{i=1}^8\cc{i}
\end{equation}
is invariant under the action of $\WEe$.

For Jacobi's elliptic function $\cd{u}$ it is well known that shifts by half periods give the following relations:
\begin{equation}
 \cd{u+2K}=-\cd{u},\quad
 \cd{u+\ii K'}=\dfrac{1}{k\, \cd{u}}.
\end{equation}
These identities motivate our search for the transformations that are identity mappings on the $\PicX$.
Indeed, we define such transformations $\io_i$, $i=1,\dots,4$, by the following actions:
\begin{subequations}\label{eqns:action_iota}
\begin{align}
 &\io_1:(\cc{1},\dots,\cc{8},\eta,x,y)
 \mapsto\left(\cc{1}-\frac{\ii K'}{2},\dots,\cc{8}-\frac{\ii K'}{2},\eta-\frac{\ii K'}{2},\frac{1}{kx},y\right),\\
 &\io_2:(\cc{1},\dots,\cc{8},\eta,x,y)
 \mapsto\left(\cc{1}-\frac{\ii K'}{2},\dots,\cc{8}-\frac{\ii K'}{2},\eta+\frac{\ii K'}{2},x,\frac{1}{ky}\right),\\
 &\io_3:(\cc{1},\dots,\cc{8},\eta,x,y)
 \mapsto\left(\cc{1}-K,\dots,\cc{8}-K,\eta-K,-x,y\right),\\
 &\io_4:(\cc{1},\dots,\cc{8},\eta,x,y)
 \mapsto\left(\cc{1}-K,\dots,\cc{8}-K,\eta+K,x,-y\right).
\end{align}
\end{subequations}
Adding the transformations $\io_i$, we extend $\WEe$ to
\begin{equation}
 \tWEe=\lrangle{\io_1,\io_2,\io_3,\io_4}\rtimes\WEe.
\end{equation}
In general, for a function $F=F(c_i,\eta,x,y)$, we let an element
$w\in\tWEe$ act as $w.F=F(w.c_i,w.\eta,w.x,w.y)$, that is, 
$w$ acts on the arguments from the left. 
\begin{theorem}\label{theorem:birational_tWE8}
Under the birational actions \eqref{eqns:action_E8_J_para_xy} and \eqref{eqns:action_iota}, 
the generators of $\tWEe=\lrangle{\io_1,\io_2,\io_3,\io_4}\rtimes\lrangle{s_0,\dots,s_8}$ satisfy the following fundamental relations:
\begin{subequations}\label{eqns:fundamental_tWe8}
\begin{align}
 &(s_is_j)^{l_{ij}}=(\io_i\io_j)^{m_{ij}}=1,\\
 &\io_is_j=s_j\io_i,\quad i=1,2,3,4,~ j\neq 1,2,\quad
 \io_{\{1,2,3,4\}}s_1=s_1\io_{\{2,1,4,3\}},\\
 &\io_1s_2=s_2\io_1\io_2,\quad
 \io_2s_2=s_2\io_2,\quad
 \io_3s_2=s_2\io_3\io_4,\quad
 \io_4s_2=s_2\io_4,
\end{align}
\end{subequations}
where 
\begin{align}
&l_{ij}=
\begin{cases}
 1,& i=j\\
 3, &i=j-1\quad (j=2,\dots,7),\quad \text{or\hspace{0.5em}if}\quad (i,j)=(3,8),(7,0)\\
 2, &\text{otherwise},
\end{cases}\\
&m_{ij}=
\begin{cases}
 1,& i=j\\
 2, &\text{otherwise}.
\end{cases}
\end{align}
\end{theorem}
\begin{proof}
Let us define the function $g(z)$ by
\begin{equation}
 g(z)
 =\bfrac{1-\frac{\cd{\eta-\cc2}}{\cd{\eta-z}}}{1-\frac{\cd{\eta-\cc1}}{\cd{\eta-z}}}
 \bfrac{1-\frac{\cd{\eta+\cc1}}{\cd{\eta+z}}}{1-\frac{\cd{\eta+\cc2}}{\cd{\eta+z}}}
 \bfrac{1-\frac{\cd{2\eta-\frac{\cc1-\cc2}{2}}}{\cd{z+\frac{\cc1+\cc2}{2}}}}{1-\frac{\cd{2\eta+\frac{\cc1-\cc2}{2}}}{\cd{z+\frac{\cc1+\cc2}{2}}}}.
\end{equation}
The action \eqref{eqn:action_E8_J_y} can be expressed as
\begin{equation}
 \bfrac{s_2(y)-\cd{2\eta-\frac{\cc1-\cc2}{2}}}{s_2(y)-\cd{2\eta+\frac{\cc1-\cc2}{2}}}
 \bfrac{x-\cd{\eta+\cc1}}{x-\cd{\eta+\cc2}}
 \bfrac{y-\cd{\eta-\cc2}}{y-\cd{\eta-\cc1}}=g(0).
\end{equation}
By replacing the Jacobian elliptic function to the Jacobi theta-functions using \eqref{eqn:Jelliptic_JTheta}
and then applying the addition formula \eqref{eqn:add_JTH}, the function $g(z)$ can be rewritten as
\begin{equation}
 g(z)=\frac{\JTK{\eta-\cc1}\JTK{\eta+\cc2}\JTK{2\eta+\frac{\cc1-\cc2}{2}}}{\JTK{\eta-\cc2}\JTK{\eta+\cc1}\JTK{2\eta-\frac{\cc1-\cc2}{2}}},
\end{equation}
which gives
\begin{equation}
 g(z)=g(0).
\end{equation}
Therefore, the action \eqref{eqn:action_E8_J_y} can be also expressed as
\begin{equation}
 \bfrac{s_2(y)-\cd{2\eta-\frac{\cc1-\cc2}{2}}}{s_2(y)-\cd{2\eta+\frac{\cc1-\cc2}{2}}}
 \bfrac{x-\cd{\eta+\cc1}}{x-\cd{\eta+\cc2}}
 \bfrac{y-\cd{\eta-\cc2}}{y-\cd{\eta-\cc1}}=g(c_3),
\end{equation}
which leads the proof of the relation $(s_2s_3)^3=1$.
The other relations can be directly verified by using the actions \eqref{eqns:action_E8_J_para_xy} and \eqref{eqns:action_iota}.
Therefore, we have completed the proof.
\end{proof}

Now we are in a position to derive Equation \eqref{eqns:RCGeqn} from the Cremona transformations
associated with $A_0^{(1)J}$-surface.
Note that for convenience we use the following notations 
for the composition of the reflections $s_i$ and for the summation of the parameters $\cc{i}$:
\begin{align}
 &s_{i_1\cdots i_m}=s_{i_1}\dots s_{i_m},\quad i_1\cdots i_m\in\{0,\dots,8\},\\
 &\cc{j_1\cdots j_n}=\cc{j_1}+\cdots+\cc{j_n},\quad j_1\cdots j_n\in\{1,\dots,8\},
\end{align}
respectively.
Let
\begin{equation}
 \RJ1=s_{56453483706756452348321 56453483706756452348321706734830468}\io_4\io_3\io_2\io_1.
\end{equation}
The actions of $\RJ1$ on the root lattice $Q(A_0^{(1)\bot})$ and parameter space are not translational as the following shows:
\begin{align}
 &\RJ1:
 \begin{pmatrix}
 \al_0\\\al_1\\\al_2\\\al_3\\\al_4\\\al_5\\\al_6\\\al_7\\\al_8
 \end{pmatrix}
 \mapsto
\left(\begin{array}{ccccccccc}
 -1 & 0 & 0 & 0 & 0 & 0 & 0 & 0 & 0 \\
 -1 & -1 & -4 & -6 & -5 & -4 & -3 & -2 & -3 \\
 0 & 0 & 1 & 2 & 1 & 0 & 0 & 0 & 1 \\
 0 & 0 & 0 & -1 & 0 & 0 & 0 & 0 & 0 \\
 0 & 0 & 0 & 0 & -1 & 0 & 0 & 0 & 0 \\
 1 & 1 & 2 & 4 & 4 & 3 & 3 & 2 & 2 \\
 0 & 0 & 0 & 0 & 0 & 0 & -1 & 0 & 0 \\
 0 & 0 & 0 & 0 & 0 & 0 & 0 & -1 & 0 \\
 0 & 0 & 0 & 0 & 0 & 0 & 0 & 0 & -1 \\
\end{array}\right)
 \begin{pmatrix}
 \al_0\\\al_1\\\al_2\\\al_3\\\al_4\\\al_5\\\al_6\\\al_7\\\al_8
 \end{pmatrix},\\
 &\RJ1(\cc{i})=-\cc{i}+\frac{\cc{1234}-\cc{5678}}{4}-\ka,\quad i=1,\dots,4,\\
 &\RJ1(\cc{j})=-\cc{j}+\frac{\cc{1234}+3\cc{5678}}{4}-\ka,\quad j=5,\dots,8,\quad
 \RJ1(\eta)=\eta+\frac{\la}{2},
\end{align}
where $\ka$ is defined by \eqref{eqn:kappa}.
However, when the parameters take special values \eqref{eqns:condition_RCG},
the action of $\RJ1$ becomes the translational motion in the parameter subspace:
\begin{equation}\label{eqn:action_RJ1_special_para}
 \RJ1:(\gae,\gao,z_0)\mapsto(\gae,\gao,z_0+2(\gae+\gao)-2\ka),
\end{equation}
and then the action on the coordinates:
\begin{equation}
 \RJ1:(x,y)\mapsto(\tx,\ty),
\end{equation}
gives Equation \eqref{eqns:RCGeqn}.
Here, we also consider the translation 
\begin{equation}\label{eqn:def_TJ1}
 \TJ1={\RJ1}^2
\end{equation}
whose actions on the root lattice $Q(A_0^{(1)\bot})$ and parameter space are given by
\begin{align}
 &\TJ1(\al_1)=\al_1-2\de,\quad
 \TJ1(\al_5)=\al_5+\de,\\
 &\TJ1:(\cc{i},\cc{i+4},\eta)\mapsto (\cc{i}-\la,\cc{i+4}+\la+4\ka,\eta+\la-2\ka),\quad i=1,\dots,4.
 \label{eqn:action_TJ1_para}
\end{align}
The action on the coordinates:
\begin{equation}
 \TJ1:(x,y)\mapsto(\ox,\oy),
\end{equation}
gives the following elliptic Painlev\'e equation:
\begin{subequations}\label{eqns:TJ1}
\begin{align}
 &\bfrac{k\,\cd{\eta-\cc8+\ka}\oy+1}{k\,\cd{\eta-\cc7+\ka}\oy+1}
 \bfrac{\tx-\cd{\eta-\cc7+\frac{\cc{5678}}{2}+\la+\ka}}{\tx-\cd{\eta-\cc8+\frac{\cc{5678}}{2}+\la+\ka}}\notag\\
 &\qquad =\gG{\frac{\cc{5678}-2\cc5+\la}{2},\frac{\cc{5678}-2\cc6+\la}{2},\frac{\cc{5678}-2\cc7+\la}{2},\frac{\cc{5678}-2\cc8+\la}{2},\eta+\frac{\la}{2}+\ka}\notag\\
 &\qquad\qquad \times\frac{\pP{\frac{\cc{5678}-2\cc5+\la}{2},\frac{\cc{5678}-2\cc6+\la}{2},\frac{\cc{5678}-2\cc7+\la}{2},\eta+\frac{\la}{2}+\ka}{\tx,\ty}}
 {\pP{\frac{\cc{5678}-2\cc5+\la}{2},\frac{\cc{5678}-2\cc6+\la}{2},\frac{\cc{5678}-2\cc8+\la}{2},\eta+\frac{\la}{2}+\ka}{\tx,\ty}},\\
 &\bfrac{k\,\cd{\eta+\cc4+\ka}\ox+1}{k\,\cd{\eta+\cc3+\ka}\ox+1}
 \bfrac{k\,\cd{\eta-\cc3+2\la+\ka}\oy+1}{k\,\cd{\eta-\cc4+2\la+\ka}\oy+1}\notag\\
 &\qquad=\gG{\eta-\cc1+\frac{\cc{1234}}{4}+\la,\eta-\cc2+\frac{\cc{1234}}{4}+\la,\eta-\cc3+\frac{\cc{1234}}{4}+\la,\eta-\cc4+\frac{\cc{1234}}{4}+\la,\frac{\cc{5678}+2\la}{4}+\ka}\notag\\
 &\qquad\qquad \times\frac{\pP{\eta-\cc1+\frac{\cc{1234}}{4}+\la,\eta-\cc2+\frac{\cc{1234}}{4}+\la,\eta-\cc3+\frac{\cc{1234}}{4}+\la,\frac{\cc{5678}+2\la}{4}+\ka}{\frac{-1}{k\oy},\tx}}
 {\pP{\eta-\cc1+\frac{\cc{1234}}{4}+\la,\eta-\cc2+\frac{\cc{1234}}{4}+\la,\eta-\cc4+\frac{\cc{1234}}{4}+\la,\frac{\cc{5678}+2\la}{4}+\ka}{\frac{-1}{k\oy},\tx}},
\end{align}
where $\tx=\RJ1(x)$ and $\ty=\RJ1(y)$ are given by
\begin{align}
 &\bfrac{k\,\cd{\eta+\cc8-\frac{\cc{5678}}{2}}\ty+1}{k\,\cd{\eta+\cc7-\frac{\cc{5678}}{2}}\ty+1}
 \bfrac{x-\cd{\eta+\cc7}}{x-\cd{\eta+\cc8}}
 =\gG{\cc5,\cc6,\cc7,\cc8,\eta}
 \frac{\pP{\cc5,\cc6,\cc7,\eta}{x,y}}{\pP{\cc5,\cc6,\cc8,\eta}{x,y}},\\
 &\bfrac{k\,\cd{\eta-\cc4+\frac{\cc{1234}}{2}}\tx+1}{k\,\cd{\eta-\cc3+\frac{\cc{1234}}{2}}\tx+1}
 \bfrac{k\,\cd{\eta+\cc3+\frac{\cc{5678}}{2}}\ty+1}{k\,\cd{\eta+\cc4+\frac{\cc{5678}}{2}}\ty+1}\notag\\
 &\qquad=\gG{\eta+\cc1+\frac{\cc{5678}}{4},\eta+\cc2+\frac{\cc{5678}}{4},\eta+\cc3+\frac{\cc{5678}}{4},\eta+\cc4+\frac{\cc{5678}}{4},\frac{\cc{5678}}{4}}\notag\\
 &\qquad\qquad\times\frac{\pP{\eta+\cc1+\frac{\cc{5678}}{4},\eta+\cc2+\frac{\cc{5678}}{4},\eta+\cc3+\frac{\cc{5678}}{4},\frac{\cc{5678}}{4}}{\frac{-1}{k\ty},x}}
 {\pP{\eta+\cc1+\frac{\cc{5678}}{4},\eta+\cc2+\frac{\cc{5678}}{4},\eta+\cc4+\frac{\cc{5678}}{4},\frac{\cc{5678}}{4}}{\frac{-1}{k\ty},x}}.
\end{align}
\end{subequations}
Here, for conciseness, we provide the terms $\gG{\aaa1,\aaa2,\aaa3,\aaa4,b}$,
$\pQ{\aaa1,\aaa2,\aaa3,\aaa4,\aaa5,b}{X}$ and 
$\pP{\aaa1,\aaa2,\aaa3,b}{X,Y}$ in Appendix \ref{section:Appendix_A}.

\begin{remark}
The periodicity of Jacobian elliptic functions in Equations \eqref{eqns:RCGeqn} and \eqref{eqns:TJ1} allows us to take the actions \eqref{eqn:action_RJ1_special_para} and \eqref{eqn:action_TJ1_para} to be
\begin{align}
 &\RJ1:(\gae,\gao,z_0)\mapsto(\gae,\gao,z_0+2(\gae+\gao)),\\
 &\TJ1:(\cc{i},\cc{i+4},\eta)\mapsto (\cc{i}-\la,\cc{i+4}+\la,\eta+\la),\quad i=1,\dots,4,
\end{align}
respectively, without loss of generality.
\end{remark}
\begin{remark}\label{remark:projective_reduction}
The relation \eqref{eqn:def_TJ1} is the key to the reduction of Equation \eqref{eqns:TJ1} to Equation \eqref{eqns:RCGeqn}. 
Under the conditions \eqref{eqns:condition_RCG}, it gives the projective reduction of Equation \eqref{eqns:TJ1} to Equation \eqref{eqns:RCGeqn}.
(See \cite{KNT2011:MR2773334,KN2015:MR3340349} for more detail.) 
We therefore refer to Equation \eqref{eqns:RCGeqn} as the projectively-reduced equation of Equation \eqref{eqns:TJ1}. 
Conversely, we refer to Equation \eqref{eqns:TJ1} as the generic version of Equation \eqref{eqns:RCGeqn}. 
Note that the word ``generic'' here has a specific meaning. 
An equation is called generic only when it has the same number of parameters as the corresponding surface-type discrete Painlev\'e equation in Sakai's list \cite{SakaiH2001:MR1882403}.
\end{remark}
\begin{remark}
Often a discrete Painlev\'e equation can be extended up to a generic version by singularity confinement \cite{GRP1991:MR1125950,RGH1991:MR1125951}.
However, Equation \eqref{eqns:RCGeqn} has a reduced number of parameters after applying the singularity confinement criterion (see \cite{RCG2009:MR2525848}).
The reason is some of the coefficients in Equation \eqref{eqns:TJ1} become zero under the condition \eqref{eqns:condition_RCG}.
In such a case, singularity confinement cannot be used to extend the equation up to the generic version.
We note that a similar observation has also been reported for another example in \cite{HHNS2015:MR3317164}.
\end{remark}

Next, we provide a new expression for the MSY elliptic Painlev\'e equation investigated in \cite{SakaiH2001:MR1882403,MurataM2004:MR2108677,MSY2003:MR1958273,ORG2002:MR1946932,ORG2001:MR1877472}. 
Moreover, we also construct its projectively-reduced equation.

Define a translation of $\tWEe$ by
\begin{equation}
 \TJ2={\RJ2}^2,
\end{equation}
where
\begin{equation}
 \RJ2=s_{23483256457067348356452348321},
\end{equation}
whose action on the root lattice $Q(A_0^{(1)\bot})$ is given by
\begin{equation}
 \TJ2(\al_1)=\al_1-2\de,\quad
 \TJ2(\al_2)=\al_2+\de.
\end{equation}

The translation $\TJ2$ actually corresponds to the translation ${\rm dP}(A_0^{(1)})$ in \cite{MSY2003:MR1958273}. 
However, it was expressed in terms of Weierstrass elliptic functions. 
We now convert this translation to Jacobian elliptic functions by using the birational actions \eqref{eqns:action_E8_J_para_xy}. 
Consider the following action of $\TJ2$: 
{\allowdisplaybreaks
\begin{subequations}\label{eqns:MSY_dP}
\begin{equation}
 \TJ2:
 \left(\begin{matrix}
 \cc1,\cc2,\cc3,\cc4\\
 \cc5,\cc6,\cc7,\cc8
 \end{matrix},\eta,x,y\right)
 \mapsto
 \left(\begin{matrix}
 \cc1,\cc2,\cc3,\cc4\\
 \cc5,\cc6,\cc7,\cc8
 \end{matrix},\eta+\la,\ox,\oy\right).
\end{equation}
Here, $\ox$ and $\oy$ are given by
\begin{align}
 &\bfrac{\oy-\cd{\eta-\cc8+\la}}{\oy-\cd{\eta-\cc7+\la}}
 \bfrac{x-\cd{\eta+\cc7}}{x-\cd{\eta+\cc8}}\notag\\
 &\quad=\bfrac{1-\frac{\cd{\eta-\cc8+\la}}{\cd{3\eta+\cc6+\la}}}{1-\frac{\cd{\eta-\cc7+\la}}{\cd{3\eta+\cc6+\la}}}
 \bfrac{1-\frac{\cd{2\eta+\frac{\cc7-\cc8}{2}+\la}}{\cd{\cc6+\frac{\cc7+\cc8}{2}-\la}}}
 {1-\frac{\cd{2\eta-\frac{\cc7-\cc8}{2}+\la}}{\cd{\cc6+\frac{\cc7+\cc8}{2}-\la}}}
 \bfrac{1-\frac{\cd{2\eta-\frac{\cc7-\cc8}{2}+\la}}{\cd{2\eta+\cc6-\frac{\cc7+\cc8}{2}+\la}}}
 {1-\frac{\cd{2\eta+\frac{\cc7-\cc8}{2}+\la}}{\cd{2\eta+\cc6-\frac{\cc7+\cc8}{2}+\la}}}
 \bfrac{1-\frac{\cd{\eta+\cc8+\frac{\cc{1234}}{2}}}{\cd{\eta+\cc5+\frac{\cc{1234}}{2}}}}
 {1-\frac{\cd{\eta+\cc7+\frac{\cc{1234}}{2}}}{\cd{\eta+\cc5+\frac{\cc{1234}}{2}}}}\notag\\
 &\qquad\times\bfrac{\frac{\pQ{\cc5,\cc6,\cc3,\cc7,\cc{1234},\eta}{x}}{\pP{\cc1,\cc2,\cc4,\eta}{x,y}}-\frac{x-\cd{\eta+\cc3}}{x-\cd{\eta+\cc4}}\,\frac{\pQ{\cc5,\cc6,\cc4,\cc7,\cc{1234},\eta}{x}}{\gG{\cc1,\cc2,\cc3,\cc4,\eta}\pP{\cc1,\cc2,\cc3,\eta}{x,y}}}
 {\frac{\pQ{\cc5,\cc6,\cc3,\cc8,\cc{1234},\eta}{x}}{\pP{\cc1,\cc2,\cc4,\eta}{x,y}}-\frac{x-\cd{\eta+\cc3}}{x-\cd{\eta+\cc4}}\,\frac{\pQ{\cc5,\cc6,\cc4,\cc8,\cc{1234},\eta}{x}}{\gG{\cc1,\cc2,\cc3,\cc4,\eta}\pP{\cc1,\cc2,\cc3,\eta}{x,y}}},\\
 &\bfrac{\ox-\cd{\eta+\cc1+\la}}{\ox-\cd{\eta+\cc2+\la}}
 \bfrac{\oy-\cd{\eta-\cc2+\la}}{\oy-\cd{\eta-\cc1+\la}}\notag\\
 &\quad=\bfrac{1-\frac{\cd{\eta+\cc1+\la}}{\cd{3\eta-\cc3+3\la}}}{1-\frac{\cd{\eta+\cc2+\la}}{\cd{3\eta-\cc3+3\la}}}
 \bfrac{1-\frac{\cd{2\eta+\frac{\cc1-\cc2}{2}+2\la}}{\cd{\cc3+\frac{\cc1+\cc2}{2}}}}
 {1-\frac{\cd{2\eta-\frac{\cc1-\cc2}{2}+2\la}}{\cd{\cc3+\frac{\cc1+\cc2}{2}}}}
 \bfrac{1-\frac{\cd{2\eta-\frac{\cc1-\cc2}{2}+2\la}}{\cd{2\eta-\cc3+\frac{\cc1+\cc2}{2}+2\la}}}
 {1-\frac{\cd{2\eta+\frac{\cc1-\cc2}{2}+2\la}}{\cd{2\eta-\cc3+\frac{\cc1+\cc2}{2}+2\la}}}
 \bfrac{1-\frac{\cd{\eta-\cc1+\frac{\cc{1234}}{2}+\la}}{\cd{\eta-\cc4+\frac{\cc{1234}}{2}+\la}}}
 {1-\frac{\cd{\eta-\cc2+\frac{\cc{1234}}{2}+\la}}{\cd{\eta-\cc4+\frac{\cc{1234}}{2}+\la}}}\notag\\
 &\qquad\times\bfrac{\frac{\pQ{\frac{\la}{2}-\cc4,\frac{\la}{2}-\cc3,\frac{\la}{2}-\cc6,\frac{\la}{2}-\cc2,\cc{1234},\eta+\frac{\la}{2}}{\oy}}{\pP{\frac{\la}{2}-\cc8,\frac{\la}{2}-\cc7,\frac{\la}{2}-\cc5,\eta+\frac{\la}{2}}{\oy,x}}-\frac{\oy-\cd{\eta-\cc6+\la}}{\oy-\cd{\eta-\cc5+\la}}\,\frac{\pQ{\frac{\la}{2}-\cc4,\frac{\la}{2}-\cc3,\frac{\la}{2}-\cc5,\frac{\la}{2}-\cc2,\cc{1234},\eta+\frac{\la}{2}}{\oy}}{\gG{\frac{\la}{2}-\cc8,\frac{\la}{2}-\cc7,\frac{\la}{2}-\cc6,\frac{\la}{2}-\cc5,\eta+\frac{\la}{2}}\pP{\frac{\la}{2}-\cc8,\frac{\la}{2}-\cc7,\frac{\la}{2}-\cc6,\eta+\frac{\la}{2}}{\oy,x}}}
 {\frac{\pQ{\frac{\la}{2}-\cc4,\frac{\la}{2}-\cc3,\frac{\la}{2}-\cc6,\frac{\la}{2}-\cc1,\cc{1234},\eta+\frac{\la}{2}}{\oy}}{\pP{\frac{\la}{2}-\cc8,\frac{\la}{2}-\cc7,\frac{\la}{2}-\cc5,\eta+\frac{\la}{2}}{\oy,x}}-\frac{\oy-\cd{\eta-\cc6+\la}}{\oy-\cd{\eta-\cc5+\la}}\,\frac{\pQ{\frac{\la}{2}-\cc4,\frac{\la}{2}-\cc3,\frac{\la}{2}-\cc5,\frac{\la}{2}-\cc1,\cc{1234},\eta+\frac{\la}{2}}{\oy}}{\gG{\frac{\la}{2}-\cc8,\frac{\la}{2}-\cc7,\frac{\la}{2}-\cc6,\frac{\la}{2}-\cc5,\eta+\frac{\la}{2}}\pP{\frac{\la}{2}-\cc8,\frac{\la}{2}-\cc7,\frac{\la}{2}-\cc6,\eta+\frac{\la}{2}}{\oy,x}}}.
\end{align}
\end{subequations}
}

Now consider the special choices of parameters: 
\begin{equation}
 \cc5+\cc4=\cc6+\cc3=\cc7+\cc2=\cc8+\cc1=\frac{\la}{2}.
\end{equation}
Then, the transformation $\RJ2$ acts on this parameter subspace and variables $x$, $y$ as
\begin{equation}
 \RJ2:
 \left(\cc1,\cc2,\cc3,\cc4,\eta,x,y\right)
 \mapsto
 \left(\cc1,\cc2,\cc3,\cc4,\eta+\frac{\la}{2},\tx,x\right).
\end{equation}
This leads to a scalar second-order mapping for $x$ alone, given by 
\begin{align}\label{eqn:MSY_dP_PR}
 &\bfrac{\ttx-\cd{\eta+\cc1+\la}}{\ttx-\cd{\eta+\cc2+\la}}
 \bfrac{\tx-\cd{\eta-\cc2+\la}}{\tx-\cd{\eta-\cc1+\la}}\notag\\
 &\quad=\bfrac{1-\frac{\cd{\eta+\cc1+\la}}{\cd{3\eta-\cc3+3\la}}}{1-\frac{\cd{\eta+\cc2+\la}}{\cd{3\eta-\cc3+3\la}}}
 \bfrac{1-\frac{\cd{2\eta+\frac{\cc1-\cc2}{2}+2\la}}{\cd{\cc3+\frac{\cc1+\cc2}{2}}}}
 {1-\frac{\cd{2\eta-\frac{\cc1-\cc2}{2}+2\la}}{\cd{\cc3+\frac{\cc1+\cc2}{2}}}}
 \bfrac{1-\frac{\cd{2\eta-\frac{\cc1-\cc2}{2}+2\la}}{\cd{2\eta-\cc3+\frac{\cc1+\cc2}{2}+2\la}}}
 {1-\frac{\cd{2\eta+\frac{\cc1-\cc2}{2}+2\la}}{\cd{2\eta-\cc3+\frac{\cc1+\cc2}{2}+2\la}}}
 \bfrac{1-\frac{\cd{\eta-\cc1+\frac{\cc{1234}}{2}+\la}}{\cd{\eta-\cc4+\frac{\cc{1234}}{2}+\la}}}
 {1-\frac{\cd{\eta-\cc2+\frac{\cc{1234}}{2}+\la}}{\cd{\eta-\cc4+\frac{\cc{1234}}{2}+\la}}}\notag\\
 &\qquad\times\bfrac{\frac{\pQ{\frac{\la}{2}-\cc4,\frac{\la}{2}-\cc3,\cc3,\frac{\la}{2}-\cc2,\cc{1234},\eta+\frac{\la}{2}}{\tx}}{\pP{\cc1,\cc2,\cc4,\eta+\frac{\la}{2}}{\tx,x}}-\frac{\tx-\cd{\eta+\cc3+\frac{\la}{2}}}{\tx-\cd{\eta+\cc4+\frac{\la}{2}}}\,\frac{\pQ{\frac{\la}{2}-\cc4,\frac{\la}{2}-\cc3,\cc4,\frac{\la}{2}-\cc2,\cc{1234},\eta+\frac{\la}{2}}{\tx}}{\gG{\cc1,\cc2,\cc3,\cc4,\eta+\frac{\la}{2}}\pP{\cc1,\cc2,\cc3,\eta+\frac{\la}{2}}{\tx,x}}}
 {\frac{\pQ{\frac{\la}{2}-\cc4,\frac{\la}{2}-\cc3,\cc3,\frac{\la}{2}-\cc1,\cc{1234},\eta+\frac{\la}{2}}{\tx}}{\pP{\cc1,\cc2,\cc4,\eta+\frac{\la}{2}}{\tx,x}}-\frac{\tx-\cd{\eta+\cc3+\frac{\la}{2}}}{\tx-\cd{\eta+\cc4+\frac{\la}{2}}}\,\frac{\pQ{\frac{\la}{2}-\cc4,\frac{\la}{2}-\cc3,\cc4,\frac{\la}{2}-\cc1,\cc{1234},\eta+\frac{\la}{2}}{\tx}}{\gG{\cc1,\cc2,\cc3,\cc4,\eta+\frac{\la}{2}}\pP{\cc1,\cc2,\cc3,\eta+\frac{\la}{2}}{\tx,x}}},
\end{align}
which is the projectively-reduced equation of the MSY elliptic Painlev\'e equation \eqref{eqns:MSY_dP}.

\begin{remark}\label{remark:differenceTj1Tj2}
Recall that a Kac translation $T_\al:\PicX\to \PicX$ is defined by
\begin{equation}
 T_\al(\la)=\la+(\de|\la)\al-\left(\frac{(\al|\al)(\de|\la)}{2}+(\al|\la)\right)\de,\quad
 \la\in \PicX,
\end{equation}
and, therefore, its squared length is given by $-(\al|\al)$.

In \cite{ORG2001:MR1877472}, Ohta {\it et al.} discovered that NVs are equivalent to each other under the operations of the Weyl group $\WEe=\langle s_0,\dots,s_8\rangle$
and the same holds for NNVs.
$\TJ1$ and $\TJ2$ are expressed by Kac transformations as the following:
\begin{equation}
 \TJ1=T_{2H_0-\sum_{i=5}^8E_i},\quad
 \TJ2=T_{H_0-H_1}.
\end{equation}
Since 
$$-\left(2H_0-\sum_{i=5}^8E_i\left|2H_0-\sum_{i=5}^8E_i\right)\right.=4,\quad -(H_0-H_1|H_0-H_1)=2,$$
$\TJ1$ and $\TJ2$ are a NNV and a NV, respectively.
In this sense, the translation $\TJ1$ differs from $\TJ2$.
\end{remark}

Remark \ref{remark:differenceTj1Tj2} shows that Equation \eqref{eqns:TJ1} forms a new system of elliptic difference equations, which differs from the MSY elliptic difference equation given originally by Sakai. 
Since the generic versions of equations are different, Equation \eqref{eqn:MSY_dP_PR} differs from Equation \eqref{eqns:RCGeqn}.
In this sense, Equation \eqref{eqn:MSY_dP_PR} is also a new  elliptic difference equation. 

\section{Concluding remarks}
\label{ConcludingRemarks}
In this paper, we have constructed the birational actions of the Cremona isometries $\tWEe$ for the $A_0^{(1)J}$-surface. 
Equation \eqref{eqns:RCGeqn} is realized in terms of birational actions of Cremona isometries for the first time.
Using birational actions, we further derived the elliptic Painlev\'e equations \eqref{eqns:TJ1}, \eqref{eqns:MSY_dP} and \eqref{eqn:MSY_dP_PR}.

Equation \eqref{eqns:MSY_dP} is a new expression of the important elliptic Painlev\'e equation studied in \cite{SakaiH2001:MR1882403,MurataM2004:MR2108677,MSY2003:MR1958273,ORG2002:MR1946932,ORG2001:MR1877472}. 
For a long time, this equation remained the unique elliptic Painlev\'e equation known with 8 parameters. 
In this paper, we provided a new elliptic Painlev\'e equation which also has 8 parameters. 
We showed that it must be different because it is realized in terms of translations in  $\tWEe$, which are not conjugate to those that give rise to the previously known elliptic difference equation.
Furthermore, we provided a projectively-reduced version of the MSY elliptic Painlev\'e equation, namely 
Equation \eqref{eqn:MSY_dP_PR}.
We also showed that Equation \eqref{eqns:TJ1} is a generic version of Equation \eqref{eqns:RCGeqn}. 

These results lead to more questions. 
The search for Lax pairs of these new equations remains open. 
Moreover, special function solutions of the projectively-reduced equations \eqref{eqn:MSY_dP_PR} and \eqref{eqns:RCGeqn}  remain to be found.  
\subsection*{Acknowledgment}
The authors would like to express their sincere thanks to Profs M. Noumi, Y. Ohta, T. Takenawa and Y. Yamada for inspiring and fruitful discussions.
This research was supported by an Australian Laureate Fellowship \# FL120100094 and grant \# DP160101728 from the Australian Research Council and JSPS KAKENHI Grant Number JP17J00092.
Nakazono gratefully acknowledges support from Erwin Schr\"{o}dinger International Institute for Mathematics and Physics (ESI) that enabled his attendance at the conference on Elliptic Hypergeometric Functions in Combinatorics, Integrable Systems and Physics where this work was presented in March 20-24, 2017.
\appendix
\section{Jacobian elliptic functions}
\label{section:Appendix_A}
To make this paper self-contained, we recall needed information about Jacobian elliptic functions and define notations (see \cite[Chapter 22]{NIST:DLMF} and \cite{book_WW1996:course}).

The Jacobian elliptic functions are given by
\begin{equation}
 \sn{u}=\sn{u,k},\quad
 \cn{u}=\cn{u,k},\quad
 \dn{u}=\dn{u,k},\quad
 \cd{u}=\frac{\cn{u}}{\dn{u}},
\end{equation}
where the functions ${\rm sn}(u)$ and ${\rm cd}(u)$ are even, while the functions ${\rm cn}(u)$ and ${\rm dn}(u)$ are odd.
The Jacobi theta-functions are given by
\begin{equation}
 H(u)=H(u,k),\quad
 \Th(u)=\Th(u,k),\quad
 H_1(u)=H(u+K),\quad
 \Th_1(u)=\Th(u+K),
\end{equation}
where the functions $\Th(u)$, $H_1(u)$ and $\Th_1(u)$ are even, while the function $H(u)$ is odd
(see \cite[\S 21.62]{book_WW1996:course} and \cite[\S 16.31]{book_AS2012:handbook} for definitions of these standard functions).
Here, $k$ and $k'$ are the modulus and the complementary modulus of the elliptic sine, respectively, which satisfy 
\begin{equation}
 k^2+k'^2=1.
\end{equation}
$K=K(k)$ and $K'=K'(k)$ are complete elliptic integrals and also periodical parameters of Jacobian elliptic functions, e.g.
the periods of the function ${\rm sn}(u)$ are $4K$ and $2\ii K'$.
We note that the following formulae hold:
\begin{align}
 &\sn{u}=\frac{H(u)}{k^{1/2}\Th(u)},\quad
 \cn{u}=\frac{k'^{1/2}H_1(u)}{k^{1/2}\Th(u)},\quad
 \dn{u}=\frac{k'^{1/2}\Th_1(u)}{\Th(u)},\label{eqn:Jelliptic_JTheta}\\
 &\ssn{u}+\ccn{u}=1,\quad
 k^2\ssn{u}+\ddn{u}=1,\quad
 \cd{u}=\sn{u+K},\\
 &\sn{u+v}=\frac{\sn{u}^2-\sn{v}^2}{\sn{u}\cn{v}\dn{v}-\sn{v}\cn{u}\dn{u}},
 \label{eqn:add_Jsn}\\
 &\cd{u+4K}=\cd{u},\quad
 \cd{u+2\ii K'}=\cd{u},\\
 &\cd{u+2K}=-\cd{u},\quad
 \cd{u+\ii K'}=\dfrac{1}{k\,\cd{u}},\\
 &\Th(u)H(v)+H(u)\Th(v)
 =\frac{2\JH{\frac{u+v}{2}}\JT{\frac{u+v}{2}}\JHK{\frac{u-v}{2}}\JTK{\frac{u-v}{2}}}{H_1(0)\Th_1(0)}.
 \label{eqn:add_JTH}
\end{align}

Throughout this paper, we also use the functions 
$\gG{\aaa1,\aaa2,\aaa3,\aaa4,b}$,
$\pQ{\aaa1,\aaa2,\aaa3,\aaa4,\aaa5,b}{X}$ and 
$\pP{\aaa1,\aaa2,\aaa3,b}{X,Y}$
defined by
{\allowdisplaybreaks
\begin{align}
 &\gG{\aaa1,\aaa2,\aaa3,\aaa4,b}
 =\bcfrac{1-\frac{\cd{\aaa4+\frac{\aaa1+\aaa2}{2}}}{\cd{\aaa2+\frac{\aaa1+\aaa2}{2}}}}
 {1-\frac{\cd{\aaa3+\frac{\aaa1+\aaa2}{2}}}{\cd{\aaa2+\frac{\aaa1+\aaa2}{2}}}}
 \bcfrac{1-\frac{\cd{b-\aaa4}}{\cd{b-\aaa1}}}{1-\frac{\cd{b-\aaa3}}{\cd{b-\aaa1}}}
 \bcfrac{1-\frac{\cd{b+\aaa4-\frac{\aaa1+\aaa2+\aaa3+\aaa4}{2}}}{\cd{b+\aaa2+\frac{\aaa1+\aaa2+\aaa3+\aaa4}{2}}}}
 {1-\frac{\cd{b+\aaa3-\frac{\aaa1+\aaa2+\aaa3+\aaa4}{2}}}{\cd{b+\aaa2+\frac{\aaa1+\aaa2+\aaa3+\aaa4}{2}}}}
 \bcfrac{1-\frac{\cd{\aaa3+\frac{\aaa1+\aaa2}{2}}}{\cd{2b+\aaa2-\frac{\aaa1+\aaa2}{2}}}}
 {1-\frac{\cd{\aaa4+\frac{\aaa1+\aaa2}{2}}}{\cd{2b+\aaa2-\frac{\aaa1+\aaa2}{2}}}},\\
 &\pQ{\aaa1,\aaa2,\aaa3,\aaa4,\aaa5,b}{X}\notag\\
 &=\left(\cd{b+\aaa3-\frac{\aaa5}{2}}-\cd{b+\aaa2+\frac{\aaa5}{2}}\right)
 \left(\cd{b+\aaa1+\frac{\aaa5}{2}}-\cd{b+\aaa4+\frac{\aaa5}{2}}\right)\notag\\
 &\quad\times\Big(\cd{b+\aaa4}\cd{b+\aaa1}+\cd{b+\aaa2}X\Big)
 +\left(\cd{b+\aaa3-\frac{\aaa5}{2}}-\cd{b+\aaa1+\frac{\aaa5}{2}}\right)\notag\\
 &\quad\times\left(\cd{b+\aaa4+\frac{\aaa5}{2}}-\cd{b+\aaa2+\frac{\aaa5}{2}}\right)
 \Big(\cd{b+\aaa4}\cd{b+\aaa2}+\cd{b+\aaa1}X\Big)\notag\\
 &-\left(\cd{b+\aaa3-\frac{\aaa5}{2}}-\cd{b+\aaa4+\frac{\aaa5}{2}}\right)
 \left(\cd{b+\aaa1+\frac{\aaa5}{2}}-\cd{b+\aaa2+\frac{\aaa5}{2}}\right)\notag\\
 &\quad\times\Big(\cd{b+\aaa1}\cd{b+\aaa2}+\cd{b+\aaa4}X\Big),\\
 &\pP{\aaa1,\aaa2,\aaa3,b}{X,Y}=C_1 X Y+C_2 X+C_3 Y+C_4,
\end{align}
where
\begin{align*}
 C_1=&\Big(\cd{b-\aaa3}-\cd{b-\aaa2}\Big)\cd{b+\aaa1}
 +\Big(\cd{b-\aaa1}-\cd{b-\aaa3}\Big)\cd{b+\aaa2}\notag\\
 &+\Big(\cd{b-\aaa2}-\cd{b-\aaa1}\Big)\cd{b+\aaa3},\\
 C_2=&\Big(\cd{b-\aaa2}-\cd{b-\aaa3}\Big)\cd{b-\aaa1}\cd{b+\aaa1}\notag\\
 &+\Big(\cd{b-\aaa3}-\cd{b-\aaa1}\Big)\cd{b-\aaa2}\cd{b+\aaa2}\notag\\
 &+\Big(\cd{b-\aaa1}-\cd{b-\aaa2}\Big)\cd{b-\aaa3}\cd{b+\aaa3},\\
 C_3=&\Big(\cd{b+\aaa3}-\cd{b+\aaa2}\Big)\cd{b-\aaa1}\cd{b+\aaa1}\notag\\
 &+\Big(\cd{b+\aaa1}-\cd{b+\aaa3}\Big)\cd{b-\aaa2}\cd{b+\aaa2}\notag\\
 &+\Big(\cd{b+\aaa2}-\cd{b+\aaa1}\Big)\cd{b-\aaa3}\cd{b+\aaa3},\\
 C_4=&\Big(\cd{b+\aaa2}\cd{b-\aaa3}-\cd{b-\aaa2}\cd{b+\aaa3}\Big)\cd{b-\aaa1}\cd{b+\aaa1}\notag\\
 &+\Big(\cd{b+\aaa3}\cd{b-\aaa1}-\cd{b-\aaa3}\cd{b+\aaa1}\Big)\cd{b-\aaa2}\cd{b+\aaa2}\notag\\
 &+\Big(\cd{b+\aaa1}\cd{b-\aaa2}-\cd{b-\aaa1}\cd{b+\aaa2}\Big)\cd{b-\aaa3}\cd{b+\aaa3}.
\end{align*}
}
\section{Relation between the $A_0^{(1)}$- and $A_0^{(1)J}$-surface}
\label{section:Appendix_B}
In this section, we give the one-to-one correspondence between the $A_0^{(1)}$- and $A_0^{(1)J}$-surface, 
which are characterized by the base points \eqref{eqn:basepoints_W} and \eqref{eqn:basepoints}, respectively.

Using Landen transformations:
\begin{equation}
 \cn{u,k}=\frac{1-\frac{2}{1+k}\sn{\frac{1+k}{2}x,\frac{2k^{1/2}}{1+k}}^2}{\dn{\frac{1+k}{2}x,\frac{2k^{1/2}}{1+k}}},\quad
 \dn{u,k}=\frac{1-\frac{2k}{1+k}\sn{\frac{1+k}{2}x,\frac{2k^{1/2}}{1+k}}^2}{\dn{\frac{1+k}{2}x,\frac{2k^{1/2}}{1+k}}},
\end{equation}
and defining Weierstrass $\wp$ function in terms of the Jacobian elliptic function ${\rm sn}$ by
\begin{equation}
 \we{u;k}=\frac{e_1(k)-e_3(k)}{\ssn{(e_1(k)-e_3(k))^{1/2}u,k}}+e_3(k),
\end{equation}
where, $e_i(k)$, $i=1,2,3$, 
satisfy
\begin{equation}
 k=-\frac{(e_2(k)-e_3(k))^{1/2}}{(e_1(k)-e_3(k))^{1/2}},\quad
 k'=\frac{(e_1(k)-e_2(k))^{1/2}}{(e_1(k)-e_3(k))^{1/2}},\quad
 e_1(k)+e_2(k)+e_3(k)=0,
\end{equation}
we obtain the following one-to-one correspondence between Jacobian elliptic function ${\rm cd}$ and Weierstrass $\wp$ function:
\begin{equation}
 \cd{u,k}=\frac{(1+k)\we{\omega;l}-2 e_1(l)+(1-k)e_3(l)}
 {(1+k)\we{\omega;l}-2k e_1(l)-(1-k)e_3(l)},
\end{equation}
where 
\begin{equation}
 \omega=\frac{(1+k)u}{2(e_1(l)-e_3(l))^{1/2}},\quad
 l=\frac{2k^{1/2}}{1+k}.
\end{equation}
By defining new coordinates $(f,g)$ and parameters $b_i$, $i=1,\dots,8$, and $t$ as the following:
\begin{subequations}\label{eqns:def_fgbt}
\begin{align}
 &x=\frac{(1+k)f-2 e_1(l)+(1-k)e_3(l)}{(1+k)f-2k e_1(l)-(1-k)e_3(l)},\quad
 y=\frac{(1+k)g-2 e_1(l)+(1-k)e_3(l)}{(1+k)g-2k e_1(l)-(1-k)e_3(l)},\\
 &b_i=\frac{(1+k)\cc{i}}{2(e_1(l)-e_3(l))^{1/2}},\quad i=1,\dots,8,\quad
 t=\frac{(1+k)\eta}{2(e_1(l)-e_3(l))^{1/2}},
\end{align}
\end{subequations}
the eight base points \eqref{eqn:basepoints} can be rewritten as
\begin{equation}\label{eqn:basepoints_W}
 p_i:(f,g)=\big(\we{b_i+t},\we{t-b_i}\big),\quad i=1,\dots,8,
\end{equation}
which is the Weierstrass $\wp$ function's setting ($A_0^{(1)}$-surface) investigated in \cite{MSY2003:MR1958273}.
Therefore, the relations \eqref{eqns:def_fgbt} give the one-to-one correspondence between the $A_0^{(1)}$- and $A_0^{(1)J}$-surface.
\def\cprime{$'$} \def\cprime{$'$}

\end{document}